\documentclass[aps,reprint,twoside,floatfix]{revtex4-1}
\usepackage[bookmarks = true, pdfpagemode = None, pdfstartview = FitH, colorlinks = true, urlcolor = blue]{hyperref}

\usepackage{graphicx}

\usepackage{amsmath}
\usepackage{amsthm}
\usepackage{bbm}

\newcommand{\hide}[1]{}

\newcommand{\bra}[1]{\langle #1 |}
\newcommand{\ket}[1]{| #1 \rangle}
\newcommand{\braket}[2]{\langle #1 | #2 \rangle}

\def \id   {\mathbbm{1}}

\newtheorem{theorem}{Theorem} 
\newtheorem{lemma}[theorem]{Lemma}

\begin{document}

\title{Preparing projected entangled pair states on a quantum computer}

\author{Martin Schwarz}
\author{Kristan Temme}
\author{Frank Verstraete}
\affiliation{Vienna Center for Quantum Science and Technology, Faculty of Physics,
University of Vienna, 
Vienna, Austria}

\begin{abstract}
We present a quantum algorithm to prepare injective PEPS on a quantum
computer, a class of open tensor networks representing quantum states.
The run-time of our algorithm scales polynomially with the inverse of
the minimum condition number of the PEPS projectors and, essentially, 
with the inverse of the spectral gap of the PEPS' parent Hamiltonian. 
\end{abstract}

\maketitle


  Projected Entangled Pair States, or PEPS \cite{VC04}, have been proposed as a class of
  quantum states especially suited to describe the ground states of local
  Hamiltonians in quantum many-body physics. PEPS are a higher-dimensional
  generalization of the one-dimensional Matrix Product States \cite{RO97}, or MPS, for
  which many interesting properties have been proven: For example, MPS provably
  approximate the ground state of 1D local Hamiltonians with constant spectral
  gap \cite{VC06,hastings07}, exhibit an area law \cite{hastings07} as well as
  an exponential decay of two-point correlation functions.
  Furthermore, for each MPS with
  the \emph{injectivity} property \cite{PVCW07},
  a parent Hamiltonian can be constructed with this MPS as its unique ground state.
  MPS can also be prepared efficiently on a quantum computer \cite{SSVCW05}.
  PEPS however form a much richer class of states, and can e.g. represent critical systems and systems with topological quantum order \cite{VWPC06}. It is conjectured that all ground states of gapped local Hamiltonians in higher dimensions can be represented faithfully as PEPS, and although there are strong indications for this fact, this has not been proven. What is clear, however, is the fact that one can also construct parent Hamiltonians for them \cite{PVCW07}, and the PEPS will be the unique ground states of those Hamiltonians if the PEPS obeys the so-called injectivity condition \cite{PVCW07}. Many physically relevant classes of PEPS on lattices are known to be almost always injective, including e.g. the 2D AKLT state \cite{PVCW07}. 
  A particularly interesting subclass of PEPS is the one that consists of all those states whose parent Hamiltonian have a gap that scales at most as an inverse polynomial as a function of the system size: in that case, a local observable (i.e. the local Hamiltonian) allows to distinguish the state from all other ones, as the ground state always has energy zero by construction. It was an open problem \cite{VWPC06} whether such states could however be even created on a quantum computer, as an algorithm that would allow to prepare any PEPS would allow for the solution of $PP$-complete problems~\cite{SWVC07}.

  In this article we show how well-conditioned injective PEPS can be prepared on a
  quantum computer efficiently. The key idea of our approach is to grow the
  PEPS step by step. We demand that not only our final PEPS is the unique ground
  state of its parent local Hamiltonian, but also that there exists a sequence of
  partial sums of the local terms of the parent Hamiltonian, such that each partial
  sum has a unique ground state of its own. Based on this assumption, the algorithm
  starts with a physical realization of the valence bond pairs as its initial state
  and iteratively performs entangling measurements on the virtual particles to map
  virtual degrees of freedom to physical ones, just as in the definition of
  the PEPS. The PEPS is called injective, iff this map is (left) invertible which
  can only be the case if the dimension of the physical space is actually at least
  as large as the dimension of the virtual space at each vertex.
  Preparing a PEPS by measurements may seem to require post-selection to
  project onto the right measurement outcome. To overcome this issue
  we use the Marriott-Watrous trick \cite{MW05,NWZ09} of undoing a measurement based
  on Jordan's lemma \cite{jordan1875} and combine it with the uniqueness property of
  injective PEPS \cite{PVCW07} to prepare the required eigenstates. 
  A key element that contributes to the success of this algorithm is the fact that the measurements are not done locally, such as in the framework of dissipative quantum state engineering \cite{VWC09}, but globally by running a phase estimation algorithm that singles out the ground subspace; a similar approach was used in the context of the quantum Metropolis sampling algorithm~\cite{TOVPV10}. 
  Alternatively, methods for eigenpath traversal \cite{ATS07,BKS10} can also be applied~\cite{BoixoEigen}.

  \hide{
  The paper is structured as follows. 
  Section \emph{Definitions and Results} reviews essential definitions and states the main theorem.
  Section \emph{Algorithm} presents the algorithm, while
  sections \emph{Bounding the transition probabilities} and \emph{Bounding the convergence rate} provide
  detailed analyses and finally proof the main theorem.
  }

  {\it Definitions and Results.}---Before stating the result, 
  we review the definition of PEPS and their essential properties.
  Recall \cite{VC04,PVCW07} that PEPS are quantum states defined over an arbitrary graph $G=(V,E)$
  such that quantum systems of local dimension $d$ are assigned to each vertex.
  We construct the PEPS by assigning to each edge $e \in E$ a maximally entangled
  state $\sum_{i=1}^D \ket{ii}$. In this way, a vertex $v \in V$ with degree $k$
  gets associated with $k$ \emph{virtual} $D$-dimensional systems. Finally, a map $A^{(v)}: \mathbbm{C}^D \otimes \mathbbm{C}^D \otimes \cdots \mathbbm{C}^D \mapsto \mathbbm{C}^d$
  is applied to each vertex, taking the $k$ \emph{virtual} $D$-dimensional systems
  to a single \emph{physical} $d$-dimensional system. The linear map $A^{(v)}$ is usually called
  the PEPS ``projector'' and is parameterized by tensors $A^{(v)}_i$ as follows:
  $A^{(v)} = \sum_{i=1}^d \sum_{j_1,\dots,j_{k}=1}^D A^{(v)}_{i,j} \ket{i}\bra{j_1,\dots,j_{k}}$
  where $A^{(v)}_i$ is a tensor with $k$ indices. The PEPS can now be written as
  $  \ket{\psi} = \sum_{i_1,\dots,i_n=1}^d \mathcal{C}[\{A^{(v)}_{i_v}\}_v] \ket{i_1,\dots,i_n}$
  where $\mathcal{C}$ means the contraction of all tensors $A_i^{(v)}$ according to
  the edges of the graph. In the most general case the virtual index dimension $D$ as
  well as the physical index dimension $d$ may also depend on the edges $e$ and vertices $v$
  of the interaction graph, but we suppress this detail in favor of simplicity.
  Note, that w.l.o.g. $A^{(v)} \geq 0$ may be assumed, since for
  arbitrary $\tilde{A}^{(v)}$ we can choose a local basis by performing
  a polar decomposition, i.e.
  $\tilde{A}^{(v)} = U^{(v)} A^{(v)}$
  with $U^{(v)}$ unitary and $A^{(v)} \geq 0$.

  A PEPS $\ket{\psi}$ is called \emph{injective} \cite{PVCW07},
  if each PEPS projector $A^{(v)}$ has a left inverse.
  For some PEPS this may only be true, after some local
  contractions of a constant number of PEPS tensors $A^{(v)}$ according
  to the interaction graph of the PEPS forming new projectors $\hat{A}^{(v)}$
  for which the condition above holds.
  Since this blocking can be performed efficiently for constant
  degree graphs, we may assume for the remainder of this paper, that
  it has already been performed, such that each individual $A^{(v)}$
  in our input is already injective by itself. Note, that the existence of a left
  inverse allows us to strengthen the assumption $A^{(v)}\geq 0$
  w.l.o.g. to $A^{(v)} > 0$ for all $v$.

  For injective PEPS, there is a simple construction \cite{PVCW07}
  of a $2$-local parent Hamiltonian, such that the injective PEPS is its
  unique, zero-energy ground state. This construction gives a parent
  Hamiltonian for a quantum system consisting of $n$ particles with
  $d$-dimensional Hilbert spaces.
  
  Let $H$ be a Hermitian matrix with $\lambda_0 < \lambda_1$ its smallest and second
  smallest eigenvalues. Then we call $\Delta(H)=\lambda_1 - \lambda_0$ the \emph{spectral gap} of $H$.
  For any matrix $A$, the \emph{condition number} $\kappa(A)$ is defined as
  $\kappa(A)=\frac{\sigma_{\max}(A)}{\sigma_{\min}(A)}$,
  where $\sigma_{\max}(A)$ and $\sigma_{\min}(A)$ are the largest and
  smallest singular values of $A$, respectively.
  We are now in a position to state the performance of our algorithm as our main theorem:
  \begin{theorem} \label{thm:main}
  Let $G=(V,E)$ be an interaction graph with bounded degree and some total order
  defined on $V$. Let $\{A^{(v)}\}_{v \in V_{[t]}}$ be a set of injective PEPS
  projectors of dimension $d \times D^k$ associated with each $v$ in $V$ up to vertex $t$ (according to the
  total vertex order) describing a sequence of PEPS $\ket{\psi_t}$, and let $\kappa=\underset{v \in V}{\max}\;\kappa(A^{(v)})$ be
  the largest condition number of all PEPS projectors.  Let $\Delta=\min_t \Delta(H_t)$,
  where $\Delta(H_t)$ is the spectral gap of the parent Hamiltonian $H_t$ of the PEPS $\ket{\psi_t}$.
  Then there exists a quantum algorithm generating the final PEPS $\ket{\psi_{|V|}}$
  with probability at least $1-\varepsilon$ in time 
  ~$\tilde{O}( \frac{|V|^2 |E|^{2}\kappa^2}{ \varepsilon \Delta} +|V|kd^6)$.
  \end{theorem}

  \begin{figure}[h]
  \raggedright
  \noindent\textbf{Input:} Interaction graph $G=(V,E)$ with degree bound $k$
  and total vertex order. For each $v \in V$, $d \times D^k$-matrices $A^{(v)}$ as PEPS projectors.
  Acceptable probability of failure $\varepsilon$.\\
  \noindent\textbf{Output:} PEPS $\ket{\psi}$ with probability at least $1-\varepsilon$.
  \begin{enumerate}
  \item $t \leftarrow 0$  	
  \item $\ket{\psi_t} \leftarrow$ entangled pair for each edge $e \in E$. \label{step:pairs}
  \item $H_t = \sum_{e \in E} H_e$ \label{step:pairterms}
  \item For each $v \in V$ according to total order:\label{step:vertexloop}
  \begin{enumerate}
  	\item $H_{t+1} \leftarrow H_{t}$ \label{step:cpH}
	\item For each neighbor $v' \in V$ of $v$:
	\begin{itemize}
		\item remove term $H_t^{(v,v')}$ from $H_{t+1}$ \label{step:rmprojprojectors}	  	
		\item compute parent Hamiltonian term $H_{t+1}^{(v,v')}$ using $A^{(v)}$
		\item add term $H_{t+1}^{(v,v')}$ to $H_{t+1}$ 
	\end{itemize}
    	\item Add $H_{phy}^{(v)}=c(\id-P_{phy}^{(v)})$ to $H_{t+1}$\label{step:addprojector}
  	\item $\ket{\psi_{t+1}^{(\perp)}} \leftarrow$ measure $H_{t+1}$ on $\ket{\psi_t}$ \label{step:measurenext}
  	\item While measured energy nonzero: \label{step:failedloop}
	\begin{enumerate}
  		\item $\ket{\psi_{t}^{(\perp)}} \leftarrow$ measure $H_{t}$ on $\ket{\psi_{t+1}^{(\perp)}}$\label{step:measureundo}
  		\item $\ket{\psi_{t+1}^{(\perp)}} \leftarrow$ measure $H_{t+1}$ on $\ket{\psi_t^{(\perp)}}$\label{step:measureredo}
	\end{enumerate}
  	\item $t \leftarrow t+1$
  \end{enumerate}
  \end{enumerate}
  \caption{Algorithm constructing injective PEPS}
  \label{fig:alg}
  \end{figure}

  {\it Algorithm.---}Conceptually, PEPS are constructed by first preparing
  entangled pair states $\ket{\psi}=\sum_i \ket{ii}$ for each edge of the
  interaction graph describing the PEPS, and then projecting the $k$ \emph{virtual}
  indices associated with each vertex to a single \emph{physical} index.
  While this construction is usually considered only a theoretical device,
  the proposed algorithm is indeed simulating the above construction for the
  case of injective PEPS with gapped Hamiltonians. This entails making the
  virtual indices physical as well.

  Figure \ref{fig:alg} presents our algorithm in pseudo-code. We proceed by explaining
  each step in detail.
  PEPS construction starts in step \ref{step:pairs} by distributing
  maximally entangled states of the desired bond dimension according
  to the interaction graph $G=(E,V)$. The resulting system is the zero-energy
  ground state of a simple Hamiltonian $H_0$ consisting purely of terms projecting
  onto $H_e=\id - \frac{1}{d}\sum_{i,j=1}^{d} \ket{ii}\bra{jj}$
  for each edge of the interaction graph (step \ref{step:pairterms}).
  Note, that this simple Hamiltonian is gapped.
  
  \begin{figure*}[t]
    \begin{center}
      \leavevmode
      \includegraphics[width=0.8\textwidth]{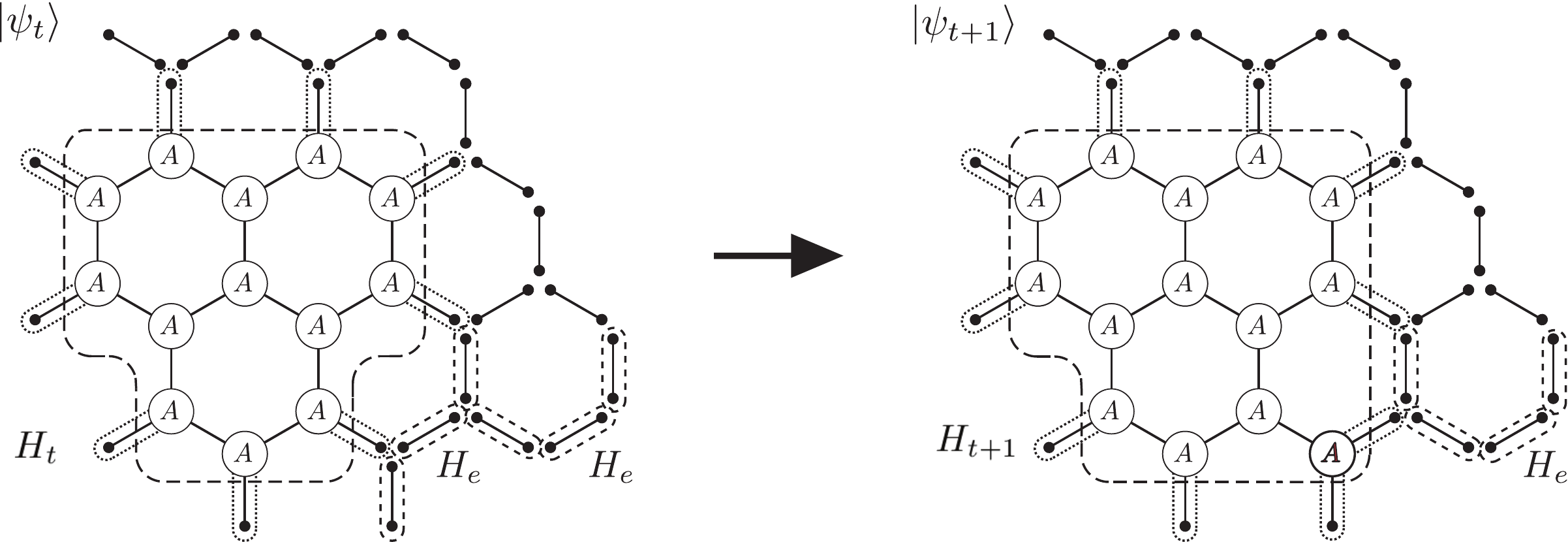}
    \end{center}
    \caption{In each step, the algorithm processes one vertex $v$: 
    $H_{t}$ is grown into $H_{t+1}$ by removing all existing terms 
       referring to~$v$, before the $2k$-local parent Hamiltonian terms are added implementing PEPS projector $A^{(v)}$ at $v$. Note, that terms around $v$ connecting to an open ``virtual'' index $v'$
       (bonds with dotted border) are only temporary and are removed in later steps of the algorithm.
       All terms constraining a vertex $v$ only restrict the physical subspace $P_{phy}^{(v)}$, while degrees of
       freedom in the orthogonal subspace $(\id-P_{phy}^{(v)})$ are eliminated with an additional penalty term $H_{phy}^{(v)}$ that is added per vertex.
       }
    \label{fig:awesome_image}
  \end{figure*}

  We now describe the main iteration of the algorithm (step \ref{step:vertexloop}),
  which is illustrated in figure \ref{fig:awesome_image}.
  In steps \ref{step:cpH}-\ref{step:addprojector}, after having selected
  the next vertex $v$ of the interaction graph according to the total vertex order,
  we construct a new Hamiltonian $H_{t+1}$ from
  $H_{t}$: First, we select a $d$-dimensional ``physical'' subspace from the $D^k$ dimensional
  space at each vertex $v$. This subspace is represented by projector $P_{phy}^{(v)}$.
  Then we remove for each neighboring vertex $v'$ of $v$ the term $H_{t}^{(v,v')}$. These are
  either trivial $H_e$ terms or temporary boundary terms (see below).
  Next we compute the new parent Hamiltonian terms $H_{t+1}^{(v,v')}$ according to \cite{PVCW07}
  and add  them to $H_{t+1}$ reflecting the application of $A^{(v)}$. 
  Restricted to the ``physical" subspace $P_{phy}^{(v)}$, each
  $H^{(v,v')}$ is simply a sum of $2$-local terms over all edges $e$ from $v$ to vertices $v'$.
  Note, that parent Hamiltonian terms $H^{(v,v')}$ towards any open ``virtual index'' $v'$
  are only temporary boundary terms which are computed in exactly the same way just as those for any
  other vertex by assuming the identity as the applicable PEPS map. Since the identity is
  trivially invertible, each intermediate PEPS is also
  injective and thus the unique ground state of the intermediate Hamitonian $H_{t+1}$.
  Since the ``physical" $d$-dimensional space is just a subspace of
  the $D^k$ ``virtual" space \emph{that is in fact also implemented physically in this algorithm},
  $H^{(v)}$ is actually a sum of $2k$-local projectors.
  In order to ensure we produce a state with a single $d$-dimensional
  local space associated to each vertex $v$ in the final PEPS, we add an extra term
  $H_{phy}^{(v)}=c(\id-P_{phy}^{(v)})$ in this step. 
  This term penalizes the orthogonal complement of the chosen subspace with
  some energy $c \gg \Delta$.

  Note, that prior to the execution of step \ref{step:measurenext}, the system is in the
  ground state $\ket{\psi_t}$ of $H_t$ by construction.
  This ground state is unique by the injectivity assumption we make for each
  intermediate PEPS $\ket{\psi_t}$ prepared in each iteration.
  In order to transition to
  the ground state $\ket{\psi_{t+1}}$ of $H_{t+1}$, we run the phase estimation \cite{NC00}
  algorithm for Hamiltonian $H_{t+1}$, perform a binary measurement to
  project $\ket{\psi_t}$
  onto the zero/non-zero energy subspaces of $H_{t+1}$, and uncompute the
  phase estimation (step \ref{step:measurenext}). This step requires an
  inverse eigenvalue gap $\Delta^{-1}$ between these two subspaces that scales with $O(poly(|V|))$
  for the phase estimation to be efficient and precise enough \cite{BACS07}. We assume that such a gap
  exists for each intermediate parent Hamiltonian $H_t$ that we construct
  according to the total vertex order defined on the interaction graph.

  If the measurement results in the projection onto the zero-energy subspace of
  $H_{t+1}$ we proceed to the next iteration (step \ref{step:failedloop}).
  By Lemma \ref{lem:overlap}, this event
  occurs with probability at least $\kappa(A^{(v)})^{-2}$, where $\kappa(A^{(v)})$ is the
  condition number of PEPS projector $A^{(v)}$ associated with vertex $v$.
  Note, that the injectivity property of the PEPS assures, that each
  $\kappa(A^{(v)})$ is a positive constant.
  If the measurement projects onto the excited subspace of $H_{t+1}$, we \emph{undo}
  the measurement by measuring $H_{t}$ again (step \ref{step:measureundo}).
  If this second measurement results in a projection on the ground state,
  we have exactly undone the (unsuccessful) measurement of $H_{t+1}$,
  otherwise the system is in the excited subspace of $H_{t}$. In both cases the
  projection onto the groundstate of $H_{t+1}$ can now be attempted again,
  with success probabilities $\kappa(A^{(v)})^{-2}$ and $1-\kappa(A^{(v)})^{-2}$,
  respectively (step \ref{step:measureredo}). By Lemma \ref{lem:convergence},
  the inner loop will succeed in projecting onto the ground
  state of $H_{t+1}$ with probability at least $1-\frac{1}{2es}$ after at most
  $\kappa^2s$ attempts, with $s$ chosen as $s=\frac{|V|}{2e\varepsilon}$.
  Once all $|V|$ vertices have been covered, the outer loop terminates with
  the PEPS $\ket{\psi}$ in its output register with probability at least $1-\varepsilon$,
  as shown in Theorem \ref{thm:main}.


  {\it Bounding the transition probabilities.---}As a first step in our 
  analysis, we need a lower bound on the transition
  probability from $\ket{\psi_t}$ to $\ket{\psi_{t+1}}$. To this end
  we proof the following lemma.

  \begin{lemma} \label{lem:overlap}
  Let $\ket{\psi_t}=\frac{1}{\sqrt{Z_t}}\ket{A_t}$ be the normalized PEPS
  $\ket{A_t}$, where $\ket{A_t}$ is the unnormalized partial PEPS resulting from the
  contraction of PEPS projectors $A^{(v)}$ for all vertices $v$ processed
  in the algorithm up to step $t$ and let $Z_t=\braket{A_t}{A_t}$. Let
  $\ket{A_{t+1}}=A_{t+1} \ket{A_t}$ where $A_{t+1}$ is the PEPS projector
  of time step $t+1$. Then $|\braket{\psi_{t+1}}{\psi_t}|^2 \geq \frac{1}{\kappa(A_{t+1})^2}>0$.
  \end{lemma}
  \begin{proof}
  A simple calculation shows
  \begin{align}
  \braket{\psi_{t+1}}{\psi_t} & = \frac{1}{\sqrt{Z_t}} \frac{1}{\sqrt{Z_{t+1}}} \bra{A_t} A_{t+1}^\dag \ket{A_t} \\
  & \geq \frac{1}{\sqrt{Z_t}} \frac{1}{\sqrt{Z_{t+1}}}
  	\frac{\bra{A_t} A_{t+1}^\dag A_{t+1} \ket{A_t}}{\sigma_{\max}(A_{t+1})} \\
  & = \frac{1}{\sqrt{Z_t}} \frac{1}{\sqrt{Z_{t+1}}} \frac{Z_{t+1}}{\sigma_{\max}(A_{t+1})} \\
  & = \frac{1}{\sigma_{\max}(A_{t+1})} \left(\frac{Z_{t+1}}{Z_{t}}\right)^{\frac{1}{2}}
  \end{align}
  where the inequality follows from the operator inequalities $A_{t+1}\geq 0$ and 
  $\frac{A_{t+1}}{\sigma_{\max(A_{t+1})}} \leq \id$. This implies
  \begin{equation}
  \left|\braket{\psi_{t+1}}{\psi_t}\right|^2 \geq \frac{1}{\sigma_{\max}(A_{t+1})^2} \frac{Z_{t+1}}{Z_{t}}
  \label{eq:overlap}
  \end{equation}
  But
  \begin{align}
  Z_{t+1} = \bra{A_t}A_{t+1}^2\ket{A_t}
          &\geq \sigma_{\min}(A_{t+1})^2 \braket{A_t}{A_t} \\
          &= \sigma_{\min}(A_{t+1})^2 Z_t. \label{eq:Z_frac}
  \end{align}
  Thus Eq. \ref{eq:overlap} and Eq. \ref{eq:Z_frac} yield the claim
  \begin{equation}
  p=|\braket{\psi_{t+1}}{\psi_t}|^2
    \geq \left(\frac{\sigma_{\min}(A_{t+1})}{\sigma_{\max}(A_{t+1})}\right)^2
     = \frac{1}{\kappa(A_{t+1})^{2}}
  \end{equation}
  Finally, the injectivity assumption of PEPS $\ket{\psi_{t+1}}$ implies
  left invertibility of $A_{t+1}$ for each $v$, thus $\kappa(A_{t+1})$ is finite,
  therefore $p>0$.
  \end{proof}

  {\it Bounding the convergence rate.---}In this section we analyze the 
  termination probability of the loop at step
  \ref{step:failedloop}.
  \begin{lemma} \label{lem:convergence}
  Let $H_{t}, H_{t+1}$ be Hamiltonians with unique zero-energy groundstates
  $\ket{\psi_{t}}$ and $\ket{\psi_{t+1}}$, respectively. Let $s$ be a positive
  integer. If the system
  is in state $\ket{\psi_{t}}$ initially, then the measurement process
  alternatingly measuring $H_{t+1}$ and $H_{t}$ and stopping
  once $\ket{\psi_{t+1}}$ is reached, takes the system to state $\ket{\psi_{t+1}}$
  with probability at least $1-\frac{1}{2es}$ after at most $s/p$
  alternations, where  $p=|\braket{\psi_{t+1}}{\psi_t}|^2$.
  \end{lemma}

  \begin{proof}
  Let $P,Q$ be the ground state projectors of $H_{t}$ and $H_{t+1}$,
  respectively, and let \mbox{$P^\perp = \id - P$}, \mbox{$Q^\perp = \id - Q$}.
  By Jordan's Lemma, there exists an orthonormal
  basis in which the Hilbert space decomposes into (1) two-dimensional
  subspaces $S_i$ invariant under both, $P$ and $Q$, and (2) one-dimensional
  subspaces $T_j$ on which $PQ$ is either an identity- or zero-projector \cite{NWZ09}.

  Since we know that $\ket{\psi_{t}}$ and $\ket{\psi_{t+1}}$ are the
  unique 1-eigenstates of $P$ and $Q$ with overlap $\sqrt{p}$, exactly
  one $S_i$ is relevant to our analysis. This two-dimensional subspace is
  spanned by both,
  $\ket{\psi_{t}}$ and some $\ket{\psi_{t}^\perp}$, as well as
  by $\ket{\psi_{t+1}}$ and some $\ket{\psi_{t+1}^\perp}$. Among these
  four vectors, we have got the following relationships \cite{MW05}:
  \begin{align}
  \ket{\psi_{t}} &= -\sqrt{p}\ket{\psi_{t+1}} + \sqrt{1-p}\ket{\psi_{t+1}^\perp}\\
  \ket{\psi_{t}^\perp} &= \sqrt{1-p}\ket{\psi_{t+1}} + \sqrt{p}\ket{\psi_{t+1}^\perp}\\
  \ket{\psi_{t+1}} &= -\sqrt{p}\ket{\psi_{t}} + \sqrt{1-p}\ket{\psi_{t}^\perp} \\
  \ket{\psi_{t+1}^\perp} &= \sqrt{1-p}\ket{\psi_{t}} + \sqrt{p}\ket{\psi_{t}^\perp}
  \end{align}
  Considering these symmetrical relations, we see that alternating measurements
  of $H_{t}$ and $H_{t+1}$ generate a Markov process among these four
  states. Since the process terminates whenever it hits $\ket{\psi_{t+1}}$,
  the only histories which can keep the process from terminating are those with
  an initial transition $\ket{\psi_t}\rightarrow \ket{\psi_{t+1}^\perp}$
  and which then keep repeating either one of the following two pairs of transitions
  \begin{align}
  \ket{\psi_{t+1}^\perp} \rightarrow \ket{\psi_{t}} &\rightarrow \ket{\psi_{t+1}^\perp}\\
  \ket{\psi_{t+1}^\perp} \rightarrow \ket{\psi_{t}^\perp} &\rightarrow \ket{\psi_{t+1}^\perp},
  \end{align}
  which occur with probabilities $(1-p)^2$ and $p^2$, respectively.
  Thus the process terminates after at most $2m+1$ measurements with probability
  \begin{align}
  p_{\text{term}}(p,m) &= 1-(1-p)(p^2+(1-p)^2)^m.
  \end{align}
  To lower-bound this probability we
  upper-bound $p_{\text{fail}}(p,m)=1-p_{\text{term}}(p,m)$ as
  \mbox{$p_{\text{fail}}(p,m)
  	  \leq (1-p)exp(-2mp(1-p))$}
  which follows from $(1-q)^m \leq e^{-qm}$, for $0\leq q\leq 1$
  and $m\geq 0$.
Finally we choose $m$ as a multiple of $\frac{1}{p}$ and find
$p_{\text{fail}}(p,s/p) \leq \frac{1}{2es}$, which can be seen
by straightforward calculus.
\end{proof}
\begin{proof}[Proof of Theorem \ref{thm:main}\label{subsec:mainproof}]
We complete the proof of Theorem \ref{thm:main} by
using Lemma \ref{lem:convergence} for bounding the failure
probability $p_{\text{fail}}$ of the inner loop to derive a lower
bound on the success probability of the outer loop over all vertices in $V$.
That is, we have to show that
$(1-p_{\text{fail}})^{|V|} \geq 1-\varepsilon$.
Since
$(1-p_{\text{fail}})^{|V|} \geq 1-{|V|} p_{\text{fail}}$
by truncating higher-order terms from the binomial series and assuming $|V|>1$
it suffices to show
${|V|} p_{\text{fail}} \leq \varepsilon$.
Using Lemma \ref{lem:convergence}, we find the first inequality of
${|V|} p_{\text{fail}} \leq \frac{{|V|}}{2es} \leq \varepsilon$,
while the second inequality is satisfied by choosing
$s \geq \frac{{|V|}}{2e\varepsilon}$.
Thus, for the algorithm to succeed with at least probability $1-\varepsilon$ we have to
choose $m \geq \frac{s}{p} \geq \frac{{|V|}}{2pe\varepsilon}$. Since we know
from Lemma \ref{lem:overlap} that $p \geq \frac{1}{\kappa^2}$, choosing
$m \geq \frac{\kappa^2 {|V|}}{2e\varepsilon} \geq \frac{{|V|}}{2pe\varepsilon}$ suffices.
Thus the inner loop performs at most
\mbox{$2m+1 \leq \frac{\kappa^2 {|V|}}{e\varepsilon}$} measurements.
The outer loop iterates over $|V|$ vertices, thus the total number of measurements is less than
$\frac{\kappa^2 {|V|}^2}{e\varepsilon}+|V|$.
Bookkeeping of the active Hamiltonian terms in the outer loop requires a total time of $O(|V|k)$
using simple arrays as data structures, and $O(|V|kd^6)$ to compute all parent Hamiltonian terms,
both of which are dominated by the $O(|V|^2)$ time of the inner loop for small $d$.
Finally, since each phase estimation
step requires $\tilde{O}(|E|^{2}/\Delta)$ \cite{BACS07,NWZ09,NC00}, where $\tilde{O}(\cdot)$ suppresses more slowly growing factors such as $exp(\sqrt{\ln(|E|/\Delta)})$ \cite{HHL09}, we find a total runtime of
$\tilde{O}( \frac{|V|^2 |E|^{2}\kappa^2}{ \varepsilon \Delta} +|V|kd^6)$.
This completes the proof of Theorem \ref{thm:main}.
\end{proof}

{\it Conclusion.---}In this Letter we have shown how to construct quantum states described by injective
PEPS in polynomial time by first reducing the problem to the generation of a sequence of
unique ground states of certain Hamiltonians and then preparing that sequence. 
In future work we will focus on extending the class of preparable PEPS and possible
performance improvements following from the results of \cite{BKS10,BS10,SB11}.

This work has been supported by Austrian SFB project FoQuS F4014.



\begin{thebibliography}{15}%
\makeatletter
\providecommand \@ifxundefined [1]{%
 \@ifx{#1\undefined}
}%
\providecommand \@ifnum [1]{%
 \ifnum #1\expandafter \@firstoftwo
 \else \expandafter \@secondoftwo
 \fi
}%
\providecommand \@ifx [1]{%
 \ifx #1\expandafter \@firstoftwo
 \else \expandafter \@secondoftwo
 \fi
}%
\providecommand \natexlab [1]{#1}%
\providecommand \enquote  [1]{``#1''}%
\providecommand \bibnamefont  [1]{#1}%
\providecommand \bibfnamefont [1]{#1}%
\providecommand \citenamefont [1]{#1}%
\providecommand \href@noop [0]{\@secondoftwo}%
\providecommand \href [0]{\begingroup \@sanitize@url \@href}%
\providecommand \@href[1]{\@@startlink{#1}\@@href}%
\providecommand \@@href[1]{\endgroup#1\@@endlink}%
\providecommand \@sanitize@url [0]{\catcode `\\12\catcode `\$12\catcode
  `\&12\catcode `\#12\catcode `\^12\catcode `\_12\catcode `\%12\relax}%
\providecommand \@@startlink[1]{}%
\providecommand \@@endlink[0]{}%
\providecommand \url  [0]{\begingroup\@sanitize@url \@url }%
\providecommand \@url [1]{\endgroup\@href {#1}{\urlprefix }}%
\providecommand \urlprefix  [0]{URL }%
\providecommand \Eprint [0]{\href }%
\providecommand \doibase [0]{http://dx.doi.org/}%
\providecommand \selectlanguage [0]{\@gobble}%
\providecommand \bibinfo  [0]{\@secondoftwo}%
\providecommand \bibfield  [0]{\@secondoftwo}%
\providecommand \translation [1]{[#1]}%
\providecommand \BibitemOpen [0]{}%
\providecommand \bibitemStop [0]{}%
\providecommand \bibitemNoStop [0]{.\EOS\space}%
\providecommand \EOS [0]{\spacefactor3000\relax}%
\providecommand \BibitemShut  [1]{\csname bibitem#1\endcsname}%
\let\auto@bib@innerbib\@empty
\bibitem [{\citenamefont {Verstraete}\ \emph {et~al.}(2006)\citenamefont
  {Verstraete}, \citenamefont {Wolf}, \citenamefont {Perez-Garcia},\ and\
  \citenamefont {Cirac}}]{VWPC06}%
  \BibitemOpen
  \bibfield  {author} {\bibinfo {author} {\bibfnamefont {F.}~\bibnamefont
  {Verstraete}}, \bibinfo {author} {\bibfnamefont {M.}~\bibnamefont {Wolf}},
  \bibinfo {author} {\bibfnamefont {D.}~\bibnamefont {Perez-Garcia}}, \ and\
  \bibinfo {author} {\bibfnamefont {J.}~\bibnamefont {Cirac}},\ }\href@noop {}
  {\bibfield  {journal} {\bibinfo  {journal} {Physical review letters}\
  }\textbf {\bibinfo {volume} {96}},\ \bibinfo {pages} {220601} (\bibinfo
  {year} {2006})}\BibitemShut {NoStop}%
\bibitem [{\citenamefont {Verstraete}\ and\ \citenamefont
  {Cirac}(2004)}]{VC04}%
  \BibitemOpen
  \bibfield  {author} {\bibinfo {author} {\bibfnamefont {F.}~\bibnamefont
  {Verstraete}}\ and\ \bibinfo {author} {\bibfnamefont {J.I.}~\bibnamefont
  {Cirac}},\ }\href@noop {} {\bibfield  {journal} {\bibinfo  {journal}
  {Physical Review A}\ }\textbf {\bibinfo {volume} {70}},\ \bibinfo {pages}
  {060302(R)} (\bibinfo {year} {2004})},\ \Eprint
  {http://arxiv.org/abs/arXiv:quant-ph/0311130} {arXiv:quant-ph/0311130}
  \BibitemShut {NoStop}%
\bibitem [{\citenamefont {Rommer}\ and\ \citenamefont
  {{\"O}stlund}(1997)}]{RO97}%
  \BibitemOpen
  \bibfield  {author} {\bibinfo {author} {\bibfnamefont {S.}~\bibnamefont
  {Rommer}}\ and\ \bibinfo {author} {\bibfnamefont {S.}~\bibnamefont
  {{\"O}stlund}},\ }\href@noop {} {\bibfield  {journal} {\bibinfo  {journal}
  {Physical Review B}\ }\textbf {\bibinfo {volume} {55}},\ \bibinfo {pages}
  {2164} (\bibinfo {year} {1997})}\BibitemShut {NoStop}%
\bibitem [{\citenamefont {Verstraete}\ and\ \citenamefont
  {Cirac}(2006)}]{VC06}%
  \BibitemOpen
  \bibfield  {author} {\bibinfo {author} {\bibfnamefont {F.}~\bibnamefont
  {Verstraete}}\ and\ \bibinfo {author} {\bibfnamefont {J.I.}~\bibnamefont
  {Cirac}},\ }\href@noop {} {\bibfield  {journal} {\bibinfo  {journal}
  {Physical Review B}\ }\textbf {\bibinfo {volume} {73}},\ \bibinfo {pages}
  {094423} (\bibinfo {year} {2006})}\BibitemShut {NoStop}%
\bibitem [{\citenamefont {Hastings}(2007)}]{hastings07}%
  \BibitemOpen
  \bibfield  {author} {\bibinfo {author} {\bibfnamefont {M.}~\bibnamefont
  {Hastings}},\ }\href@noop {} {\bibfield  {journal} {\bibinfo  {journal}
  {Journal of Statistical Mechanics: Theory and Experiment}\ }\textbf {\bibinfo
  {volume} {2007}},\ \bibinfo {pages} {P08024} (\bibinfo {year}
  {2007})}\BibitemShut {NoStop}%
\bibitem [{\citenamefont {{Perez-Garcia}}\ \emph {et~al.}(2008)\citenamefont
  {{Perez-Garcia}}, \citenamefont {{Verstraete}}, \citenamefont {{Cirac}},\
  and\ \citenamefont {{Wolf}}}]{PVCW07}%
  \BibitemOpen
  \bibfield  {author} {\bibinfo {author} {\bibfnamefont {D.}~\bibnamefont
  {{Perez-Garcia}}}, \bibinfo {author} {\bibfnamefont {F.}~\bibnamefont
  {{Verstraete}}}, \bibinfo {author} {\bibfnamefont {J.~I.}\ \bibnamefont
  {{Cirac}}}, \ and\ \bibinfo {author} {\bibfnamefont {M.~M.}\ \bibnamefont
  {{Wolf}}},\ }\href@noop {} {\bibfield  {journal} {\bibinfo  {journal} {Quant.
  Inf. Comp}\ }\textbf {\bibinfo {volume} {8}},\ \bibinfo {pages} {0650}
  (\bibinfo {year} {2008})},\ \Eprint {http://arxiv.org/abs/arXiv:0707.2260}
  {arXiv:0707.2260} \BibitemShut {NoStop}%
\bibitem [{\citenamefont {Sch{\"o}n}\ \emph {et~al.}(2005)\citenamefont
  {Sch{\"o}n}, \citenamefont {Solano}, \citenamefont {Verstraete},
  \citenamefont {Cirac},\ and\ \citenamefont {Wolf}}]{SSVCW05}%
  \BibitemOpen
  \bibfield  {author} {\bibinfo {author} {\bibfnamefont {C.}~\bibnamefont
  {Sch{\"o}n}}, \bibinfo {author} {\bibfnamefont {E.}~\bibnamefont {Solano}},
  \bibinfo {author} {\bibfnamefont {F.}~\bibnamefont {Verstraete}}, \bibinfo
  {author} {\bibfnamefont {J.I.}~\bibnamefont {Cirac}}, \ and\ \bibinfo {author}
  {\bibfnamefont {M.}~\bibnamefont {Wolf}},\ }\href@noop {} {\bibfield
  {journal} {\bibinfo  {journal} {Physical review letters}\ }\textbf {\bibinfo
  {volume} {95}},\ \bibinfo {pages} {110503} (\bibinfo {year}
  {2005})}\BibitemShut {NoStop}%
\bibitem [{\citenamefont {Schuch}\ \emph {et~al.}(2007)\citenamefont {Schuch},
  \citenamefont {Wolf}, \citenamefont {Verstraete},\ and\ \citenamefont
  {Cirac}}]{SWVC07}%
  \BibitemOpen
  \bibfield  {author} {\bibinfo {author} {\bibfnamefont {N.}~\bibnamefont
  {Schuch}}, \bibinfo {author} {\bibfnamefont {M.}~\bibnamefont {Wolf}},
  \bibinfo {author} {\bibfnamefont {F.}~\bibnamefont {Verstraete}}, \ and\
  \bibinfo {author} {\bibfnamefont {J.I.}~\bibnamefont {Cirac}},\ }\href@noop {}
  {\bibfield  {journal} {\bibinfo  {journal} {Physical review letters}\
  }\textbf {\bibinfo {volume} {98}},\ \bibinfo {pages} {140506} (\bibinfo
  {year} {2007})}\BibitemShut {NoStop}%
\bibitem [{\citenamefont {Marriott}\ and\ \citenamefont
  {Watrous}(2005)}]{MW05}%
  \BibitemOpen
  \bibfield  {author} {\bibinfo {author} {\bibfnamefont {C.}~\bibnamefont
  {Marriott}}\ and\ \bibinfo {author} {\bibfnamefont {J.}~\bibnamefont
  {Watrous}},\ }\href@noop {} {\bibfield  {journal} {\bibinfo  {journal}
  {Computational Complexity}\ }\textbf {\bibinfo {volume} {14}},\ \bibinfo
  {pages} {122} (\bibinfo {year} {2005})}\BibitemShut {NoStop}%
\bibitem [{\citenamefont {Nagaj}\ \emph {et~al.}(2009)\citenamefont {Nagaj},
  \citenamefont {Wocjan},\ and\ \citenamefont {Zhang}}]{NWZ09}%
  \BibitemOpen
  \bibfield  {author} {\bibinfo {author} {\bibfnamefont {D.}~\bibnamefont
  {Nagaj}}, \bibinfo {author} {\bibfnamefont {P.}~\bibnamefont {Wocjan}}, \
  and\ \bibinfo {author} {\bibfnamefont {Y.}~\bibnamefont {Zhang}},\
  }\href@noop {} {\bibfield  {journal} {\bibinfo  {journal} {QIC}\ }\textbf
  {\bibinfo {volume} {9}},\ \bibinfo {pages} {1053} (\bibinfo {year} {2009})},\
  \Eprint {http://arxiv.org/abs/arXiv:0904.1549} {arXiv:0904.1549} \BibitemShut
  {NoStop}%
\bibitem [{\citenamefont {Jordan}(1875)}]{jordan1875}%
  \BibitemOpen
  \bibfield  {author} {\bibinfo {author} {\bibfnamefont {C.}~\bibnamefont
  {Jordan}},\ }\href@noop {} {\bibfield  {journal} {\bibinfo  {journal} {Bull.
  Soc. Math. France}\ }\textbf {\bibinfo {volume} {3}},\ \bibinfo {pages} {103}
  (\bibinfo {year} {1875})}\BibitemShut {NoStop}%
\bibitem [{\citenamefont {Verstraete}\ \emph {et~al.}(2009)\citenamefont
  {Verstraete}, \citenamefont {Wolf},\ and\ \citenamefont {Cirac}}]{VWC09}%
  \BibitemOpen
  \bibfield  {author} {\bibinfo {author} {\bibfnamefont {F.}~\bibnamefont
  {Verstraete}}, \bibinfo {author} {\bibfnamefont {M.}~\bibnamefont {Wolf}}, \
  and\ \bibinfo {author} {\bibfnamefont {J.}~\bibnamefont {Cirac}},\
  }\href@noop {} {\bibfield  {journal} {\bibinfo  {journal} {Nature Physics}\
  }\textbf {\bibinfo {volume} {5}},\ \bibinfo {pages} {633} (\bibinfo {year}
  {2009})},\ \Eprint {http://arxiv.org/abs/arXiv:0803.1447} {arXiv:0803.1447}
  \BibitemShut {NoStop}%
\bibitem [{\citenamefont {Temme}\ \emph {et~al.}(2011)\citenamefont {Temme},
  \citenamefont {Osborne}, \citenamefont {Vollbrecht}, \citenamefont {Poulin},\
  and\ \citenamefont {Verstraete}}]{TOVPV10}%
  \BibitemOpen
  \bibfield  {author} {\bibinfo {author} {\bibfnamefont {K.}~\bibnamefont
  {Temme}}, \bibinfo {author} {\bibfnamefont {T.}~\bibnamefont {Osborne}},
  \bibinfo {author} {\bibfnamefont {K.}~\bibnamefont {Vollbrecht}}, \bibinfo
  {author} {\bibfnamefont {D.}~\bibnamefont {Poulin}}, \ and\ \bibinfo {author}
  {\bibfnamefont {F.}~\bibnamefont {Verstraete}},\ }\href@noop {} {\bibfield
  {journal} {\bibinfo  {journal} {Nature}\ }\textbf {\bibinfo {volume} {471}},\
  \bibinfo {pages} {87} (\bibinfo {year} {2011})}\BibitemShut {NoStop}%
\bibitem [{\citenamefont {Nielsen}\ and\ \citenamefont {Chuang}(2000)}]{NC00}%
  \BibitemOpen
  \bibfield  {author} {\bibinfo {author} {\bibfnamefont {M.}~\bibnamefont
  {Nielsen}}\ and\ \bibinfo {author} {\bibfnamefont {I.}~\bibnamefont
  {Chuang}},\ }\href@noop {} {\emph {\bibinfo {title} {{Quantum Computation and
  Quantum Information}}}}\ (\bibinfo  {publisher} {Cambridge University Press,
  Cambridge},\ \bibinfo {year} {2000})\BibitemShut {NoStop}%
\bibitem [{\citenamefont {Berry}\ \emph {et~al.}(2007)\citenamefont {Berry},
  \citenamefont {Ahokas}, \citenamefont {Cleve},\ and\ \citenamefont
  {Sanders}}]{BACS07}%
  \BibitemOpen
  \bibfield  {author} {\bibinfo {author} {\bibfnamefont {D.}~\bibnamefont
  {Berry}}, \bibinfo {author} {\bibfnamefont {G.}~\bibnamefont {Ahokas}},
  \bibinfo {author} {\bibfnamefont {R.}~\bibnamefont {Cleve}}, \ and\ \bibinfo
  {author} {\bibfnamefont {B.}~\bibnamefont {Sanders}},\ }\href
  {http://dx.doi.org/10.1007/s00220-006-0150-x} {\bibfield  {journal} {\bibinfo
   {journal} {Communications in Mathematical Physics}\ }\textbf {\bibinfo
  {volume} {270}},\ \bibinfo {pages} {359} (\bibinfo {year} {2007})},\ \bibinfo
  {note} {10.1007/s00220-006-0150-x}\BibitemShut {NoStop}%
\bibitem [{\citenamefont {Harrow}\ \emph {et~al.}(2009)\citenamefont {Harrow},
  \citenamefont {Hassidim},\ and\ \citenamefont {Lloyd}}]{HHL09}%
  \BibitemOpen
  \bibfield  {author} {\bibinfo {author} {\bibfnamefont {A.}~\bibnamefont
  {Harrow}}, \bibinfo {author} {\bibfnamefont {A.}~\bibnamefont {Hassidim}}, \
  and\ \bibinfo {author} {\bibfnamefont {S.}~\bibnamefont {Lloyd}},\
  }\href@noop {} {\bibfield  {journal} {\bibinfo  {journal} {Physical review
  letters}\ }\textbf {\bibinfo {volume} {103}},\ \bibinfo {pages} {150502}
  (\bibinfo {year} {2009})}\BibitemShut {NoStop}%
\bibitem [{\citenamefont {Aharonov}\ and\ \citenamefont {Ta-Shma}(2007)}]{ATS07}%
  \BibitemOpen
  \bibfield  {author} {\bibinfo {author} {\bibfnamefont {D.}~\bibnamefont
  {Aharonov}}\ and\ \bibinfo {author} {\bibfnamefont {A.}~\bibnamefont
  {Ta-Shma}},\ }\href@noop {} {\bibfield  {journal} {\bibinfo  {journal} {SIAM
  Journal on Computing}\ }\textbf {\bibinfo {volume} {37}},\ \bibinfo {pages}
  {47} (\bibinfo {year} {2007})}\BibitemShut {NoStop}%
\bibitem [{\citenamefont {Boixo}\ \emph {et~al.}(2010)\citenamefont {Boixo},
  \citenamefont {Knill},\ and\ \citenamefont {Somma}}]{BKS10}%
  \BibitemOpen
  \bibfield  {author} {\bibinfo {author} {\bibfnamefont {S.}~\bibnamefont
  {Boixo}}, \bibinfo {author} {\bibfnamefont {E.}~\bibnamefont {Knill}}, \ and\
  \bibinfo {author} {\bibfnamefont {R.}~\bibnamefont {Somma}},\ }\href@noop {}
  {\bibfield  {journal} {\bibinfo  {journal} {Arxiv preprint arXiv:1005.3034}\
  } (\bibinfo {year} {2010})}\BibitemShut {NoStop}%
\bibitem [{\citenamefont {Boixo}\ and\ \citenamefont {Somma}(2010)}]{BS10}%
  \BibitemOpen
  \bibfield  {author} {\bibinfo {author} {\bibfnamefont {S.}~\bibnamefont
  {Boixo}}\ and\ \bibinfo {author} {\bibfnamefont {R.D.}~\bibnamefont {Somma}},\
  }\href@noop {} {\bibfield  {journal} {\bibinfo  {journal} {Physical Review
  A}\ }\textbf {\bibinfo {volume} {81}},\ \bibinfo {pages} {032308} (\bibinfo
  {year} {2010})}\BibitemShut {NoStop}%
\bibitem [{\citenamefont {Somma}\ and\ \citenamefont {Boixo}(2011)}]{SB11}%
  \BibitemOpen
  \bibfield  {author} {\bibinfo {author} {\bibfnamefont {R.}~\bibnamefont
  {Somma}}\ and\ \bibinfo {author} {\bibfnamefont {S.}~\bibnamefont {Boixo}},\
  }\href@noop {} {\bibfield  {journal} {\bibinfo  {journal} {Arxiv preprint
  arXiv:1110.2494}\ } (\bibinfo {year} {2011})}\BibitemShut {NoStop}%
\bibitem [{\citenamefont {Boixo}(2011)}]{BoixoEigen}%
  \BibitemOpen
  \bibfield  {author} {\bibinfo {author} {\bibfnamefont {S.}~\bibnamefont
  {Boixo}},\ }\href@noop {} {}\bibinfo {howpublished} {personal communication}
  (\bibinfo {year} {2011})\BibitemShut {NoStop}%
\end{thebibliography}
%

\end{document}